\documentclass[11pt]{article}
\usepackage[utf8]{inputenc}
\usepackage{amsthm}
\usepackage{amsfonts}
\usepackage{mathtools}
\usepackage{amsmath}
\usepackage{ stmaryrd }
\usepackage{ amssymb }
\usepackage{ stmaryrd }
\usepackage{mathrsfs}
\usepackage{tikz}
\usetikzlibrary{arrows,positioning,decorations,automata,backgrounds,petri,fit,calc}

\usepackage[ruled,vlined,linesnumbered]{algorithm2e} 

\usepackage[noend]{algorithmic}
\usepackage{hyperref}
\usepackage{amsmath}

\usepackage{tcolorbox}

\newtheorem{theorem}{Theorem}[section]
\newtheorem{claim}[theorem]{Claim}
\newtheorem{definition}[theorem]{Definition}

\newtheorem{observation}[theorem]{Observation}

\newtheorem{corollary}[theorem]{Corollary}
\usepackage[margin=1.2in]{geometry}

\newcommand{\negl}{\mathtt{negl} }
\newcommand{\epoch}{\mathtt{epoch} }
\providecommand{\keywords}[1]
{
  \small	
  \textbf{\textit{Keywords---}} #1
}
\newcommand{\calA}{\cal A}

\usepackage{graphicx}
\graphicspath{ {images/} }

\title{Mirror Games Against an Open Book Player\footnote{Research supported in part by grants from the Israel Science Foundation (no.\ 2686/20), by the Simons Foundation Collaboration on the Theory of Algorithmic Fairness.}}
\author{
  Roey Magen\footnote{Department of Computer Science and Applied Mathematics, Weizmann Institute of Science, Rehovot, Israel. Email: \href{mailto:roey.magen@weizmann.ac.il}{roey.magen@weizmann.ac.il}.
  } \and
  Moni Naor\footnote{Department of Computer Science and Applied Mathematics, Weizmann Institute of Science, Rehovot, Israel. Incumbent of the Judith Kleeman Professorial Chair.  Email: \href{mailto:moni.naor@weizmann.ac.il}{moni.naor@weizmann.ac.il}.
  }
}
\date{}

\begin{document}
\maketitle
\begin{abstract}
Mirror games were invented by Garg and Schneider (ITCS 2019). Alice and Bob take turns (with Alice s) in declaring numbers from the set $\{1,2, \ldots, 2n\}$. If a player picks a number that was previously played, that player loses the game and the other player wins. If all numbers are
declared without repetition, the result is a draw. Bob has a simple mirror strategy that assures he won't lose the game and requires no memory. On the other hand, Garg and Schneider showed  that every deterministic Alice
requires memory of size that is proportional to $n$ in order to secure a draw.

Regarding probabilistic strategies, previous work  showed that  assuming Alice has access to a {\em secret} random perfect matching over $\{1,2, \ldots, 2n\}$ allows her to achieve a draw in the game w.p.\ at least $1-\frac{1}{n}$ and using only polylog bits of memory.  

We show that the requirement for secret bits is crucial: for an `open book' Alice with no secrets (Bob knows her memory but not future coin flips) and memory of at most $n/4c$ bits for any $c\geq 2$, there is a Bob that wins w.p.\ close to $1-{2^{-c/2}}$.  

\end{abstract}
\section{Introduction}
In the mirror game, Alice (the first player) and Bob (the second player) take turns picking numbers belonging to the set $\{1,2, \ldots, 2n\}$. A player loses if they repeat a number that has already been picked. After $2n$ rounds, when no more numbers are left to say, then the result of the game is a draw. With perfect memory, both players can keep track of all the numbers that have been announced and avoid losing. 

Bob, who plays second, has a simple deterministic, low memory strategy (which can actually be performed even by humans): Bob fixes any perfect matching on the elements $\{1,2, \ldots, 2n\}$; for example $(1,2),(3,4),...,(2n-1,n)$.  For every number picked by Alice, Bob responds with the matched number. Unless Alice repeats a number, this allows Bob to pick in every turn a number 
that has not appeared so far. Note that Bob needs to remember simply the current value  Alice announced, i.e.\ $O(\log n)$ bits, in order to execute this strategy.

This game was suggested by Garg and Schneider~\cite{mirrorgame}
as a very simple 
example of a mirror strategy (hence the name).  Given the limited  computational  resources mirror strategies  require, the question is when can they be used. In particular whether Alice
has a lower memory strategy as well and if not what is it that breaks the symmetry between the two players. 

What Garg and Schneider 
showed is that every deterministic ``winning” (drawing) strategy
for Alice requires space that is linear in $n$.
I.e.\ there is no mirror like strategy for her.

What about probabilistic strategies? Here we have to consider the model carefully. Suppose that Alice has a secret memory of some bounded size (we call this the `closed book' case) as well as a supply of random bits (that Bob cannot access).  
Garg and Schneider showed a randomized strategy that manages to draw against any strategy of Bob  with probability at least $1-\frac{1}{n}$ and requires $O(\sqrt{n}\log^2 n)$ bits of memory while relying on access to a {\em secret random perfect matching oracle}. 
Feige \cite{feige} showed an improved randomized strategy, under the same settings, that draws with the same probability and needs only $O(\log^3 n)$ bits of memory.
Menuhin and Naor~\cite{MenuhinN22} 
showed that such strategies can be implemented with either $O(n\log n)$
bits of long lasting randomness (in addition to the $O(\log^3 n)$ bits of memory), or to replace the assumptions with cryptographic ones (assuming that Bob cannot break a one-way function) and obtain a low memory strategy for Alice.

\subsection*{An Open Book Alice}
In contrast to the above model, we consider the case where Alice has {\em no secret memory}. 
Alice  has $m<2n$ memory bits, but their content is known to Bob at any point. She  has an ``unlimited" supply of randomness: in each turn Alice can ask for any number of additional random bits, and then she has access to those bits throughout the rest of the game. However, since Alice has no secrets, Bob knows her random bits {\em after she asks for them}\footnote{This is similar to the ``full information model" in distributed computing.}. 
The question we consider is whether Bob has a strategy that forces Alice to lose with high probability when Alice has such limited memory? 

As usual, we assume that Bob is aware of Alice's strategy and creates his own strategy accordingly. The question is whether for any Alice that has sublinear memory there is a Bob that makes her lose with certain probability. As we shall see, this is indeed the case. We will first see a strategy that makes Alice lose with probability about $1/2$ (Theorem~\ref{theorem-alice-loses-wp-half}) and then show how to amplify this to any constant probability of losing (Theorem~\ref{theorem-alice-loses-whp}): any Alice that has memory of size smaller than $n/c$ loses except with probability at most $2^{-\gamma c}$ for some constant $\gamma>0$.

Note that if Alice chooses all her random bits in advance, then  a simple argument shows that she will lose with probability 1: once the random bits are chosen, then her strategy is deterministic and from the work of Garg and Schneider  there is the best choice of Bob's strategy that makes her lose. So the whole difficulty revolves around the issue of Alice choosing the random bits on the fly.

\paragraph{Outline of adversarial strategy:}
The rough outline of the strategy is for Bob to start by a sequence of random moves that do not include any of the numbers selected so far. At some midpoint Bob changes course. He selects a pair of subsets, so that one of them is the correct one played up to this point and the other one is a `decoy'. The pair is selected so that from Alice's point of view these two sets are indistinguishable. Bob continues by avoiding elements that belong to this pair of subsets. Alice has to be the first one to venture into this territory and make a guess who is the real one and who is the decoy. She loses with probability $1/2$.

\section{Combinatorial Preliminaries}
\subsection{Oddtown-Eventown}
Garg and Schneider~\cite{mirrorgame}, in showing their linear lower bound on the space complexity of any deterministic winning or drawing
strategy of Alice,
used the famous  Eventown-Oddtown Theorem of Berlekamp from extremal combinatorics. This theorem  puts a bound on the maximum number of even sets where every pair has an odd intersection: 
\begin{definition}\label{oddtown}
An “Oddtown” is a collection $\cal F$ of subsets of 
$\{1, 2, \ldots, N\}$ where every subset $s \in {\cal F}$ is of even cardinality, but the intersection of every pair of distinct subsets $s_i, s_j \in {\cal F}$ has an odd cardinality. 
\end{definition}

\begin{theorem}\label{Oddtown-theorem} 
Every Oddtown contains at most $N$ subsets. 
\end{theorem}

A proof can be found, for instance, in Fox's lecture notes~\cite{fox-notes}.
We will use the following corollary of the theorem:
\begin{corollary}
\label{pair-claim}
Let $\cal F$ be a collection of subsets of $\{1, 2, \ldots,  N\}$ where every subset $s \in F$ has an even cardinality. For every integer $j$, if $|{\cal F}| > N+j$ then $\cal F$ contains at least $\lceil \dfrac{j}{2} \rceil$ disjoint pairs of subsets from $F$,  where for each pair $(s_i, s_j)$ the union $s_i \cup s_j$ is of even cardinality. We call this a
matching $P$ of pairs.
\end{corollary}

\begin{proof}
As long as $\cal F$ contains more than $N$ elements, then by Theorem~\ref{Oddtown-theorem} it contains pair of distinct subsets that have an intersection of even cardinality. Since every subset in $\cal F$ has even cardinality, the union of the pair has even cardinality as well. Put the pair in the matching $P$, delete the pair from $\cal F$ and continue the process.
\end{proof}

\begin{claim}\label{silly-claim}
Let ${\cal A}_1, \ldots, {\cal A}_m$ be events whose union is the entire sample space. 
Let $H_b = \{{\cal A}_i|\Pr({\cal A}_i) > b, i \in \{1, \ldots,m\}\}$ be the set of events with relative ``high" probability.
Let $\calA$ be the event that one of the events from $H_b$ happens.  If $b < \frac{1}{m}$, then
   $\Pr[{\cal A}] \geq 1-(m-1)\cdot b$. 
\end{claim}
\begin{proof}
Note that the number of events with probability at most $b$ is 
\begin{equation}
       \Pr[{\cal A}] = 1 - \Pr[{\cal A}^c]
       \geq 1-(m-|H_b|)\cdot b \tag{1} \label{tag1}
\end{equation}
Since $b < \frac{1}{m}$, it is not possible that $|H_b| = 0$. 
Thus,
$$
    \Pr[{\cal A}] \geq 1-(m-|H_b|)b \geq 1-(m-1)\cdot b
$$
where the first inequality is from Inequality~(\ref{tag1}). 
\end{proof}

\begin{definition}
A  function $f\colon \mathbb {N} \to \mathbb {R}$ is {\em negligible} if for every positive integer $c$ there exists an integer $N_c$ such that for all $x > N_c$, $   |f (x)|< \frac {1}{x^c}$. We use the term $\negl$ to denote some negligible function. 
\end{definition}

\section{Introduction to the Mirror Game}

\subsection{Strategies and memory}
Alice has $m<n$ bits of memory, and we refer to her state as $x \in \{0,1\}^m$.
In addition, at the end of turn $2i$ (after Bob picks a number) Alice chooses a string of random bits $r_i \in \{0,1\}^\ell$. This string $r_i$ is added to the set of random strings she can access. I.e.\ at the $i$th  round Alice has at her disposal the memory state $x \in \{0,1\}^m$ and the collection of strings $R_A^{i} = (r_1, r_2, \ldots r_i)$. 
The size of $\ell$ can be arbitrary large.
If the value of $i$ is clear from the context, we will simply use $R_A$.

Formally, the description of Alice's strategy  consists of two functions:
\begin{enumerate}
    \item  \textbf{Next move function} $${\mathtt{ nxt}}\colon \{0,1\}^m\times \{0,1\}^{i \cdot \ell} \rightarrow \{1, 2, \ldots, 2n\}$$ receiving the current memory state, the random bits chosen in steps $r_1, \ldots, r_i$ and outputting a number  in $\{1, 2, \ldots, 2n\}$. 
    \item  \textbf{State transition function} 
    $$\mathtt{upd}\colon \{1,2 \ldots, 2n\} \times \{0,1\}^m \times \{1,2, \ldots, 2n\}\times \{0,1\}^{i \cdot \ell} \rightarrow \{0,1\}^m$$ receiving the last move of Bob, the current memory state, turn number $2i+1$, the random bits chosen in steps $r_1, \ldots, r_i$ and assigning a new memory state. 
\end{enumerate}
Note that Bob knows the functions $\mathtt{nxt}$ and $\mathtt{upd}$. Regarding  the random bits of Alice, at round $i$, when Alice receives a new sequence of random bits $r_i$, Bob has access to $r_i$ as well (but not to the future bits, before they are selected).

Let $R_A$ and $R_B$ be the random variables that represent the random bits of Alice and Bob respectively. Note that in the beginning of turn $2i+1$, $R_A^{i} = (r_1,r_2, \ldots, r_i)$ and at that point Bob knows the value of $R_A^{i}$. Let $S^{i}$ be the random variable that represents the subset of numbers that have been picked by either player until turn $i$. 
Similarity, $X_i$ is the random variable that represents the memory state of Alice in the $i$th turn. If the value of $i$  is clear, we will simply use $X$ and $S$.

Given $x \in \{0,1\}^m$, turn $ i\in\{1,2, \ldots,2n\}$ and $R_A$, for a subset $s \subseteq \{1,2, \ldots, 2n\}$  of size $i$  we will consider 
the following notation: 
\begin{itemize}
    \item $\Pr[S=s] := \Pr_{R_B}[\text{set of used numbers is } s|R_A = r_A]$
    \item $\Pr[X_i = x] := \Pr_{R_B}[X_i =x|R_A=r_A]$, and if $i$ is clear we will use $\Pr[X=x]$.
\end{itemize}
Now we will consider the conditional probability that $s$ was the subset of chosen numbers that lead to memory state $x$:
$$
    \Pr[S=s|x] :=  {\Pr}_{R_B}[\text{The set of used numbers is } s|\text{memory state is } x, R_A^i=r_A].
$$

\begin{definition}
\label{def-feasible-sets}
For turn $1 \leq i \leq 2n$, memory state $x \in \{0,1\}^m$ and random strings $r_A$, denote with $\mathcal{F}_x^{r_A}$ the collection of sets 
$s \subseteq \{1,2, \ldots 2n\}$ of size $i$  s.t.\ it is possible that the state of memory of Alice at turn $i$ is $x$ when the set of numbers played is equal to $s$ and Alice's random string is $r_A$. Call those sets the {\em feasible sets}. 
\end{definition}

For ease of notation we sometime use  $\mathcal{F}_x$, instead of  $\mathcal{F}_x^{r_A}$. 

\section{Bob's Strategy Against an Open Book Alice }

We present a strategy for Bob against any Alice with $o(n)$ memory and without secrets. Namely, Bob has access to Alice's memory state, random bits as well as her strategy, namely the state transition function $upd$ and next move function $nxt$, but not to any future random bits.
We will first see in Section~\ref{sec-half} a strategy that makes Alice lose with probability about $1/2$.
We will then see in Section~\ref{sec-amplify} how to
amplify this strategy and reduce Alice's probability of winning to any $1/2^c$ 
at the price of  limiting her memory to be at most $n/ac$ for some constant $a$.


\subsection{Bob Can Win With Probability Close to Half}
\label{sec-half}

We present a strategy for Bob that makes him win with  probability  close to half. Let $1 \leq k \leq n$ be a value that  will be determined later (we will see that $k=0.4n$). Until turn $2k$ Bob chooses a number uniformly  at random from those that have not been picked so far. 

Let $S$ be the random variable that represents the actual set of used numbers and let $x\in\{0,1\}^m$ be the memory state of Alice after turn $2k$. We will show that if Alice has low memory, then with high probability Bob can find a partition of the feasible sets $\mathcal{F}_x^{r_A}$ (Definition~\ref{def-feasible-sets}) into pairs. For every pair $(s_1, s_2)$  in the partition, (i) the two subsets $s_1$ and $s_2$ have similar conditional probabilities that they are the sets actually played given memory state $x \in \{0,1\}^m$ and randomness $r_A$, and  (ii) the cardinality of the union of the pair is even. Namely,
\begin{enumerate}
\item
    $\Pr_{R_B}[S=s_1|X= x, R_A=r_A] \approx \Pr_{R_B}[S=s_2|X=x, R_A=r_A]$.
    \item $|s_1 \cup s_2| \textrm{ is even}$.
\end{enumerate}
Note that this partition does not depend on the actual set of used numbers, just on $R_A$ and $x$. But as we show w.h.p $s$ belongs to some pair in the partition, meaning that there exists $(s_1,s_2)$ in the partition s.t. $S \in \{s_1,s_2\}$.
Bob's strategy after turn $2k$ is just to pick some number that has not been used and that is not in $s_1 \cup s_2$. 
  
We claim that since $|s_1 \cup s_2|$ is even, Alice must be the one that chooses the first value in $s_1\cup s_2$ after turn $2k$. Since $s_1$ and $s_2$ have similar probabilities, even if a ``little birdy"  tells Alice  about $s_1$ and $s_2$, she will not be able to differentiate between them in real-time. Namely, she does not know whether $S = s_1$ or $S = s_2$.
Therefore she loses with probability close to half.

\begin{tcolorbox}
\textbf{Bob's strategy for making Alice lose with probability $1/2$}
\begin{enumerate}
    \item From turn $2$ until turn $2k$: Bob chooses  uniformly at random a number that has not been picked so far.
    \item Following turn $2k$  Bob finds two feasible sets $s_1$ and $s_2$ of size $2k$ with similar conditional probabilities s.t.\ the actual set of number picked so far  $s$ is equal to one of them, and $|s_1 \cup s_2|$ is even. 
\item From turn $2k$ till the end, Bob picks some number that has not been used and that is not in $s_1 \cup s_2$. 
\end{enumerate}
\end{tcolorbox}
Note that after turn $2k$ Bob's strategy can be deterministic. For example pick the smallest number that has not been used and that is not in $s_1 \cup s_2$.   
\begin{claim}\label{claim-upper-bound}
For every subset $s \subseteq \{1,2, \ldots ,2n\}$ and for every assignment $r_A$ to the random bits of Alice, if $|s| = 2k$ then
$$\Pr_{R_B}[S=s | R_A=r_A] \leq {\left(\dfrac{2k}{2n}\right)}^k$$

\end{claim}

\begin{proof}
For every turn $i \in \{2,4, \ldots, 2k\}$, if all of the numbers picked so far belong to $s$, then the cardinality of the set of numbers that have not been used and belongs to $s$ is exactly $2k-i+1$. Thus the probability that Bob chooses a number from s is exactly $\frac{2k-i+1}{2n-i+1}$ independent from $r_A$. Hence, let $Z=\{2,4,\ldots 2k\}$ and we get:
$$
    \Pr_{R_B}[S=s|R_A=r_A] \leq \displaystyle \prod_{i\in Z}  \dfrac{2k-i+1}{2n-i+1}  
    \leq 
    \prod_{i\in Z}  \dfrac{2k-1}{2n-1} 
    =  {\left(\dfrac{2k-1}{2n-1}\right)}^k \leq {\left(\dfrac{2k}{2n}\right)}^k  
$$
\end{proof}

\begin{corollary}\label{corollary-explicit-upperbound}
For every subset $s \subseteq \{1,2,...,2n\}$ and for every assignment $r_A$ to the random bits
of Alice, if $k = 0.4n$ and $|s| = 2k$ then
$$
\Pr_{R_B}[S=s | R_A=r_A] \leq {0.4}^{0.4n} = {0.4}^{0.2\cdot2n}  \leq  {0.84}^{2n}
$$
\end{corollary}
These numbers were chosen  for convenience, where the goal is to get an  upper bound on $\Pr_{R_B}[S=s]$ of the form  $\alpha^{2n}$ for some constant $0 < \alpha <1$.

We call a memory state $x$ {\em useful} for $r_A$ if $\Pr_{R_B}[X_{2k}=x | R_A=r_A] > {0.9}^{2n}$. Once again, this number was chosen for convenience. We want to  get a lower bound on $\Pr[X=x|R_A=r_A]$ of the form $\beta^{2n}$ for some constant $0 < \beta <1$  s.t.\ $\alpha<\beta$. 
Let $U_{r_A}\subseteq \{0,1\}^m$ be the set of useful memory for $r_A$.
We argue that w.h.p.\ Alice's memory state in turn $2k$ is useful.

\begin{claim}\label{claim-usfull-whp}
Let $X$ be the random variable representing the memory state after turn $2k$. For any assignment to Alice's coin flips $r_A$, if $m < 0.2n$, then $$\Pr_{R_B}[X \in U_{r_A}|R_A=r_A] \geq  1-\negl(n).$$
\end{claim}
\begin{proof}
 Let $A_x$ be the event that $X=x$. Note that for every $x \in \{0,1\}^m$ we get that $\Pr[A_x] \leq {0.9}^{2n} \leq \frac{1}{2^m}$. Then by Claim~\ref{silly-claim}:
$$
    \Pr_{R_B}[x \in U_{r_A}]  \geq 1-(2^m-1)\cdot {0.9}^{2n}  > 1 - {0.97}^{2n} +{0.9}^{2n} = 1 - \negl(n).
$$
where the last inequality is true since $2^m0.9^{2n} \leq {2^{0.2n}}\cdot 0.9^{2n} < 0.97^{2n} $.
\end{proof}

\begin{claim}\label{claim-partition}
For any assignment $r_A$, if $x$ is useful, then Bob can find a subset $W_x^{r_A} \subseteq \mathcal{F}_x^{r_A}$ and a partition $P_x^{r_A}$ of $W_x^{r_A}$ into pairs s.t.:
\begin{enumerate}
    \item $\Pr_{R_B}[ S \in W_x^{r_A}|x\in U_{r_A},R_A=r_A] > 1 - \negl(n)$.
    \item For every pair $(s_1,s_2) \in P_x^{r_A}$ we have that  $|s_1 \cup s_2| $ is even.
    \item  For every pair $(s_1,s_2) \in P_x^{r_A}$ the ratio of their conditional probabilities is close to $1$:  $$\frac{2n}{2n+1} \leq  \dfrac{\Pr_{R_B}[S=s_1|x,R_A=r_A]}{\Pr_{R_B}[S=s_2|x,R_A=r_A]} \leq \frac{2n+1}{2n}$$
\end{enumerate}

\end{claim}

\begin{proof}
The idea is to partition $\mathcal{F}_x^{r_A}$ into layers of the possible subsets that lead to memory state $x$ according to the values of their conditional probability. The pair $(s_1,s_2)\in P_x^{r_A}$, only if $\Pr[S=s_1 |x,r_A]$ and $\Pr[S=s_2 | x,r_A]$ belong to the same layer and $|s_1 \cup s_2|$ is even. 
We consider ${(2n)}^{3}$ different layers. For $j \in \{1, \ldots, (2n)^3\}$ let the $j$th  interval of probabilities  be
$$I_{j} = \left[{\left(\dfrac{{2n}}{{2n}+1}\right)}^j  \ , \ {\left(\dfrac{2n}{2n+1}\right)}^{j-1} \right].$$
Every layer $L_j$ represents the collection 
of subsets whose conditional probability given memory state $x$ and  $r_A$ is in the  interval $I_j$. Thus,  $L_j = \{s \in \mathcal{F}_x^{r_A} | \Pr_{r_B}[S=s|X=x,r_A] \in  I_j \}$. Note that if $s_1$ and $s_2$ in $L_j$,
then $$\frac{2n}{2n+1} \leq  \frac{\Pr_{r_B}[S=s_1|x,r_A]}{\Pr_{r_B}[S=s_2|x,r_A]} \leq \frac{2n+1}{2n}.$$

Note that every subset in $L_j$ has even cardinality. Thus, if layer $L_j$ contains more then $2n$ subsets, then by Claim~\ref{pair-claim} it contains at least 
$\frac{|L_j|-N}{2}$ pairs with an even intersection. Call this set of pairs $P_j$.
Let $P_x^{r_A}$ be the collection of all pairs from $P_1, P_2, \ldots, P_{(2n)^3}$. Denote the collection of all subsets that belong to some pair in $P_x^{r_A}$ as $W_x^{r_A}$. Namely,
$$    P_x^{r_A} := \{(s_1,s_2) \in P_j|j\in [{(2n)}^3]\}, \ \ W_x^{r_A} := \{s\subseteq\{1,2,...,2n\}|\exists s',(s,s')\in P_x^{r_A}\}.$$

We will see that $P_x^{r_A}$ and $W_x^{r_A}$ meet all the three conditions of the claim:

\begin{enumerate}
    \item Since $x$ is useful and by Corollary~\ref{corollary-explicit-upperbound} we get that for every s $\in \mathcal{F}_x^{r_A}$:
$$
\Pr_{R_B}[S=s|X=x,R_A=r_A] \leq \dfrac{\Pr[S=s|R_A=r_A]}{\Pr[X=x|R_A=r_A]}\leq \dfrac{0.84^{2n}}{0.9^{2n}} < 0.934^{2n}
$$
Therefore for every $F\subseteq \mathcal{F}_x^{r_A}$, if $F$ contains at most a  polynomial number of subsets, 
then $\Pr[S \in F|X=x,r_A]$ is negligible. There are ${(2n)}^3$ layers, and in every layer $i$, there are at most $2n$ subsets that do not belong to $P_i$. Thus we get at most ${(2n)}^4$ subsets that not belong to $P_x^{r_A}$. 
In addition  define: 
$$L_x = \{s\in \mathcal{F}_x^{r_A} | \ \Pr[S=s|x,r_A] < \frac{1}{2^{4n}}\}. 
$$
Since the cardinality of $L_x$ is at most $2^{2n}$, we get that $$\Pr[S \in L_x | X=x,r_A] < 2^{2n}\cdot \frac{1}{2^{4n}} = 2^{-2n}.$$

We choose $(2n)^3$ layers, since ${\left(\frac{2n}{2n+1}\right)}^{(2n)^3-1} \leq \frac{1}{2^{4n}}$ for large enough $n$. Overall $W_x^{r_A}$ contains all of the subsets from $\mathcal{F}_x^{r_A}$, except some of the subsets in $L_x$ and at most ${(2n)}^4$ other subsets, then the claim is follows. 

\item Correct by the construction of $P_x^{r_A}$. 
\item By construction for every pair $(s_1,s_2) \in P_x^{r_A}$ there exists $j$ s.t $(s_1,s_2) \in L_j$.

\end{enumerate}
\end{proof}

\begin{observation}
\label{observation-alice-picks-first}
Since $|s_1 \cup s_2|$ is even, then Alice is the first to pick an element of $s_1\cup s_2$ after turn $2k$.
\end{observation}

Now we will show that even if from turn $2k$ and on Alice has an unlimited size memory and randomness, and even if Alice knows $W_x^{r_A}$ for every memory state $x$, she still loses with probability close to half.

\begin{theorem} \label{theorem-alice-loses-wp-half}
If Alice has no secrets, and memory of size at most $0.2n$ bits, then there is a strategy for Bob where Alice  loses with probability at least 
   $\dfrac{1}{2} \cdot \dfrac{2n}{2n+1} -\negl(n)$ 
\end{theorem}

\begin{proof}
Let $r_A$ be assignment to the random bits of Alice. If a memory state $X$ is useful for $r_A$ and $S \in W_x^{r_A}$, then there exists $(s_1,s_2)\in P_x^{r_A}$ s.t.\ $S\in (s_1,s_2)$, $|s_1 \cup s_2|$ is even and $$\frac{2n}{2n+1} \leq  \frac{\Pr[S=s_1|x]}{\Pr[S=s_2|x]} \leq \frac{2n+1}{2n}$$ 

From turn $2k$, Bob chooses a number that has not been used and not in $s_1 \cup s_2$. Then by Observation~\ref{observation-alice-picks-first} the first one after turn $2k$ to pick a number from $s_1 \cup s_2$ is Alice. Intuitively, since $s_1$ and $s_2$ have similar conditional probability given $X$ and $r_A$, Alice cannot determine whether $S=s_1$ or $S=s_2$. Then the best choice  Alice can make the first time she picks a number from $s_1 \cup s_2$, is to choose a number from the set with the lowest conditional probability. But even then, she loses w.p.\ similar to half. 

Formally, let $a$ be the first number from $s_1 \cup s_2$ that Alice picks after turn $2k$.
We want to analyze:
$$
    \Pr_{R_B}[a \in S |X \textrm{ is useful}, S\in W_X^{r_A},R_A=r_A]
$$
Note that if $a\in S$, then Alice loses.
Indeed, for any randomness $r_A$ and useful memory state $x\in U_{r_A}$ :
\begin{align*}
    &\Pr_{r_B}[a \in S |X = x, S\in \{s_1,s_2\},(s_1,s_2)\in P_x^{r_A},R_A=r_A]  \\ \tag{1} \label{eq-information}
    &\geq \min_{s'\in \{s_1,s_2\}} \Pr_{r_B}[S = s' |X = x, S\in \{s_1,s_2\},(s_1,s_2)\in P_x^{r_A},R_A=r_A] \\
    &= \min_{s'\in \{s_1,s_2\}} \dfrac{\Pr_{r_B}[S = s' |X = x,R_A=r_A]}{\Pr_{r_B}[S \in \{s_1,s_2\}|X = x,R_A=r_A]} \\
    &\geq  \min_{s'\in \{s_1,s_2\}} \dfrac{\Pr_{r_B}[S = s' |X = x,R_A=r_A]}{\Pr_{r_B}[S = s_1|X = x,R_A=r_A] + \Pr_{r_B}[S = s_2|X = x,R_A=r_A]}    \\
    &\geq 
    \dfrac{\left(\frac{2n}{2n+1}\right)^{j}}{2 \left( \frac{2n}{2n+1}\right)^{j-1}}  \tag{2} \label{eq-same-layer} \\
    &=  \dfrac{2n}{2n+1} \cdot \dfrac{1}{2} \tag{3} \label{eq-close-to-half} 
\end{align*}
Where all the probabilities are over Bob's coins $R_B$.
Inequality~\ref{eq-information} is true since given $S \in (s_1,s_2)$, without knowing S, decide if $S=s_1$ or $S = s_2$ is hard in terms of information at least as choosing number from S. Inequality~\ref{eq-same-layer} is true by property number 3 of Claim~\ref{claim-partition}.

By Claim \ref{claim-usfull-whp} the memory state $x$ in turn $2k$ is useful with probability at least $1-\negl(n)$. By Claim~\ref{claim-partition}(first property) $S \in W_x^{r_A}$ with probability of at least $1-\negl(n)$. Finally, by Equation~\ref{eq-close-to-half} we get that Alice loses w.p.\ at least
$    \dfrac{2n}{2n+1} \cdot \dfrac{1}{2} -\negl(n). $


\end{proof}

\subsection{Amplifying  the Probability that Bob Wins}
\label{sec-amplify}

For simplicity, we assume that $n$ is even. Let $c \geq 2$ s.t.\ $n/c$ is even as well. We will show  a strategy for Bob where he wins with  probability about $1-{(\frac{1}{2})}^{c/2}$ against any open book Alice with bounded memory. The idea is to use similar arguments to the previous section, but instead of finding only one pair of subsets with similar conditional probability and an even union, we would like to find a collection of such pairs.
We define $c/2$  breakpoints in time where each breakpoint represents some even numbered turn. Each breakpoint and memory state $x$, define a new partition to pairs of subsets (where the subsets are possible one that could have been used since the last breakpoint). The subsets are of size $n/c$ and the matched subsets  should have similar conditional probabilities. 

For  $1 \leq t \leq c/2$ and a set $s$ of the numbers used between breakpoint $t-1$ and breakpoint $t$, we will show that w.h.p.\ $s$ belongs to some pair of subsets $(s_1^t,s_2^t)$ in the partition of feasible subsets for the given memory and Alice's coins; as before, $s_1^t$ and $s_2^t$ have similar conditional probability and an even union. Then Bob adds this pair to the collection of pairs that he already has from previous breakpoints. He uses this collection of pairs in order to define his winning strategy against  Alice. Let us formalize  the foregoing discussion. 

\textbf{Breakpoint} number $t$ for $t\in \{0,1, \ldots,c/2\}$ is chosen at turn number $k_t := t \cdot \frac{n}{c}$. Since $n/c$ is even then $k_t$ is even.

The span of turns between two successive breakpoints is called an $\epoch$. In particular, for $t\in\{1, \dots, c/2\}$, the span of $n/c$ turns between turns $k_{t-1}+1$ and $k_{t}$ is called a $t-\epoch$ (note that the epoch ends at the $t$th breakpoint.)

At the end of each breakpoint $t \in \{1, \ldots, c/2\}$, Bob comes up with new pair $(s_1^t,s_2^t)$. We will show later how Bob finds these pairs. Denote with $C_t = (s_1^1,s_2^1),...,(s_1^{t},s_2^{t})$ the collection of pairs that Bob comes up until breakpoint number $t$.  
Let $C_0 := \emptyset$ and $C_t := C_{t-1} \cup \{(s_1^t,s_2^t)\}$.

\begin{tcolorbox}

\textbf{Bob's amplified strategy:}

 {\bf Phase 1}
\begin{itemize}
\item \textbf{From turn $2$ Until turn $n/2$}: 
during the $t$th $\epoch$ Bob chooses {\em uniformly at random} among the elements that do not belong to any subset $s$ that is part of a pair in $C_{t-1}$.  Namely, from 
$$ \{1, \ldots, 2n\}\backslash \bigcup \{s| s \in \{s_1^j, s_2^j\} \text{ and } (s_1^j,s_2^j)\in C_{t-1}\}$$

\item Let $C = C_{c/2} = \{(s_1^1,s_2^1),...,(s_1^{c/2},s_2^{c/2})\}$ be the collection of pairs of subsets from the first $n/2$ turns. 
For any turn  $n/2 < i \leq 2n$, let 
$C_i^*$ be all pairs $(s_1^t,s_2^t) \in C$ such that no member of  $s_1^t\cup s_2^t$ was played from 
breakpoint $t$ until turn $i$. 
\end{itemize}

\noindent
{\bf Phase 2}
\begin{itemize}

\item \textbf{ For every turn $i \in \{n/2+2,n/2+4,...,2n\}$}, Bob picks  some number that has not been used and not in
$$\bigcup \{s| s \in \{s_1^j, s_2^j\} \text{ and } (s_1^j,s_2^j)\in C^*_i\}$$
\end{itemize}
\end{tcolorbox}

We still need to specify how Bob sets the pairs $(s_1^t,s_2^t)$ in each breakpoint $t$, and what is the meaning of similar conditional probability in this case.

\begin{figure}[th]
\fbox{\includegraphics[width=14cm,height=3.2cm]{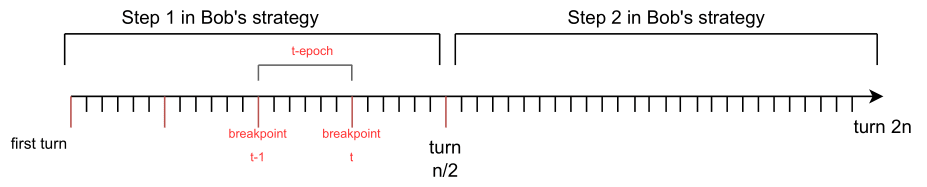}}
\caption{The red lines represent the breakpoints, the black lines represent the turns throughout the game. Not to scale.}
\centering
\end{figure}

Let $X^t$ be the random variable that represents the memory state of Alice at  breakpoint number $t$, 
let $S^t$ be the r.v.\ that represents the set of numbers that have been used during the  $t$th $\epoch$, let $R_A^t$ be the r.v.\ that represents the randomness of Alice until breakpoint number $t$ and
let $H^t$ be the r.v.\ that represent all the history of the game until breakpoint number $t$.
Namely,

  $$H^t := (S^1,...,S^t,X^1,...,X^t,C^t)$$

Note that $H_t$ represents the set of numbers that were played, the memory states of Alice in each breakpoint and the pairs that Bob chooses, until break point number t.

Given history $h^{t-1}$ and randomness $r_A$. Denote with $\mathcal{F}_{x^t}^{h^{t-1},r_A}$ the collection of feasible sets 
$s \subseteq \{1, \ldots, 2n\}$ of size $n/c$  s.t.\ it is possible that the set of numbers picked during the $t$th epoch was $s$ given all the history $h^{t-1}$ until breakpoint $t-1$ and Alice's random string in breakpoint number $t$ is $r_A$.

For simplicity, we use the notation $\mathcal{F}_x^t$, instead of $\mathcal{F}_{x^t}^{h^{t-1},r_A^t}$. We want to find a partition into pairs $P_t$ of $\mathcal{F}_x^t$ that includes most of the subsets in $\mathcal{F}_x^t$. 
First we show an upper bound for any set $s$ on the probability that the set  played in the $(t)-\epoch$ is $s$. The number of values from $\{1, \ldots 2n\}$ that are not in $C_{t-1}$ is at least $n$. Therefore,  in each turn during  epoch $t$, Bob has at least 
$n$ numbers from which he chooses one  at random one in Phase 1 of his strategy. If size of $s$ is $n/c$, then, similarly to Claim~\ref{claim-upper-bound}, we get
$$    \Pr_{R_B}[S^t = s|R_A=r_A,H^{t-1}=h^{t-1}] \leq {\left(\dfrac{n/c}{n}\right)}^{n/2c} = {\left(\dfrac{1}{c}\right)}^{2n/4c}$$ 
Let $\alpha = \left(\frac{1}{c}\right)^{1/4c}$. Note that $\alpha < 1$. 
Let $\alpha < \beta< 1$.  We call a memory state $x$ \textbf{useful} if $\Pr_{R_B}[X^t = x|H^{t-1}=h^{t-1},R_A=r_A] > {\beta}^{2n}$. 
We want to show that w.h.p.\ Alice's memory state in each breakpoint is useful. 
Choose $\delta \in (0,1)$ s.t.\ $2^{\delta} \cdot \beta < 1$. Recall that the motivation to define a useful memory state is to have an upper bound on the conditional probability of any set of numbers to be the one used in the $t-\epoch$ over Bob's coins.

Let $U_t^{r_A}\subseteq \{0,1\}^m$ be the set of useful memory states.
\begin{claim}\label{claim-usefull-whp-2}

Let $x$ be the memory state after turn $k_t$ (breakpoint $t$). For every randomness $r_A$ and history $h^{t-1}$, if $m < \delta \cdot 2n$, then $\Pr_{R_B}[ X \in U_t^{r_A}|R_A=r_A,H^{t-1} = h^{t-1}] \geq  1-\negl(n)$.
\end{claim}
\begin{proof}

By Claim~\ref{silly-claim}:
$$
    \Pr[x  \in U_t^{R_A}|R_A=r_A,H^{t-1} = h^{t-1}]  \geq 1-(2^m-1)\cdot {\beta}^{2n} 
    > 1 - {(2^{\delta}\beta)}^{2n} +{\beta}^{2n} 
$$
which is $ 1 - \negl(n)$
\end{proof}
Recall that the motivation to define a useful memory state is to have an upper bound on the conditional probability of set of used numbers during $(t)-\epoch$ over Bob's coins.
In our case, $ \Pr[S^t = s|X^t = x,H^{t-1}=h^{t-1},R^A=r_A,x\in U_t^{r_A}]$ is bounded from above by  $a^{2n}$, for some constant $a$ that is  smaller than one. If $x$ is useful we get, 
\begin{align*}
    &\Pr[S^t = s|X^t = x,H^{t-1}=h^{t-1},R_A=r_A] \\ 
    &\leq \frac{\Pr[S^t = s|R_A=r_A, H^{t-1}=h^{t-1}]}{\Pr[X^t = x|R_A=r_A, H^{t-1}=h^{t-1}]} \\
    &\leq (\frac{\alpha}{\beta})^{2n}
\end{align*}
Thus Claim~\ref{claim-partition} holds in this case as well. Namely,
\begin{claim} \label{claim-partition-2}
For any assignment $r_A$. If $x$ is useful, then Bob can find a subset $W_x^t\subseteq \mathcal{F}_x^{t}$ and a partition $P_x^t$ of $W_x^t$ into pairs s.t:
\begin{enumerate}
    \item $\Pr[S^t \in W_x^t|X^t = x,x\in U_t^{r_A},H^{t-1}=h^{t-1}] > 1 - \negl(n)$
    \item For every pair $(s_1,s_2) \in P_x^t$, $|s_1 \cup s_2| $ is even.
    \item For every pair $(s_1,s_2) \in P_x^t$: 
    $$\frac{2n}{2n+1} \leq  \dfrac{\Pr[S^t = s_1|X^t = x, x \in U_t^{r_A},H^{t-1}=h^{t-1}]}{\Pr[S^t = s_2|X^t = x, x \in U_t^{r_A},H^{t-1}=h^{t-1}]} \leq \frac{2n+1}{2n}$$
\end{enumerate}
\end{claim}

Claim~\ref{claim-partition-2} gives us the  appropriate definition of similar probability for two subsets. Therefore, in every breakpoint $t$ and for a given set  $s$ of values picked during $t-\epoch$, Bob choose  the pair $(s_1^t,s_2^t)$ s.t.\ $s \in (s_1^t,s_2^t)$ and $(s_1^t,s_2^t) \in P_x^t$. This pair is then added to $C_{t-1}$ to get $C_t$. 

\begin{theorem} \label{theorem-alice-loses-whp}
If Alice has no secrets, and memory of size of at most $\delta \cdot 2n$ bits, then there exists a strategy for Bob such that against it Alice loses with probability of at least
   $$ 1 - \left(  \dfrac{n+1}{2n+1} \right)^{c/2} - \negl(n)$$ 
\end{theorem}

\begin{proof}
If memory state $x^t$ is useful and $S^t \in W_x^t$, then Bob can find a pair $(s_1^t,s_2^t) \in P_x^t$ s.t.\ $S^t \in (s_1^t,s_2^t)$, $|s_1^t \cup s_2^t|$ is even and $$\frac{2n}{2n+1} \leq  \frac{\Pr[s_1^t|x,H^{t-1}=h^{t-1}]}{\Pr[s_2^2|x,H^{t-1}=h^{t-1}]} \leq \frac{2n+1}{2n}.$$ 

From breakpoint $t$ and on, Bob chooses a number that has not been used and not in $C_t$.
Moreover, at every round in epoch $1 \leq t' \leq c/2$ 
there are at least $n$ numbers that have not been used and are not in $C_{t'}$. Thus Bob has probability of picking a specific number is at most $1/n$. 

As in  Observation~\ref{observation-alice-picks-first}, for every  $1 \leq t \leq c/2$, the first one to pick a number from $s_1^t \cup s_2^t$ after breakpoint $t$ is Alice, since Bob avoids picking from these sets until Alice does so (and by parity she will be the first to do so). Call the element chosen  $a_t$ and the number of turn $q_t$ (note that the $q_t$'s are not necessarily increasing in $t$; furthermore, $q_t$ might be smaller than $n/2$). Therefore, as in Theorem~\ref{theorem-alice-loses-wp-half}, Eq.~\ref{eq-close-to-half} and applying Claim~\ref{claim-partition-2} we get that,
\begin{align*}
      &\Pr[a_t \in S^t |X^t \in U_t, S^t\in \{s_1^t,s_2^t\},(s_1^t,s_2^t)\in P_x^t,H^{t-1}=h^{t-1},R_A=r_A]  \\
      &\geq \min_{s'\in \{s_1^t,s_2^t\}} \Pr[S^t = s' |X^t \in U_t, S^t\in \{s_1,s_2\},(s_1,s_2)\in P_x^t,H^{t-1}=h^{t-1},R_A=r_A] \\ &\geq  \dfrac{2n}{2n+1} \cdot \dfrac{1}{2} \tag{5} \label{eq-close-to-half-2}
\end{align*}
In turn $q_t$ Alice needs to determine whether $S^t = s_1^t$ or $S^t = s_2^t$. We claim that even if Alice already guessed (before turn $q_t$) the values of $S_{j_1},...,S_{j_{\ell}}$ for breakpoints $j_1,...,j_{\ell}$  different from $t$, she still has the same probability of guessing $S^t$ correctly as in Eq.~\ref{eq-close-to-half-2}. The reason is that $s_1^t$ and $s_2^t$ have similar conditional probability given all of the history until breakpoint $t-1$. Therefore knowledge of those guesses does not give her any advantage at turn $q_t$.

By Claim~\ref{claim-usefull-whp-2} the memory state $x^t$ in breakpoint $t$ is useful with a probability of at least $1-\negl(n)$.
By Claim~\ref{claim-partition-2} $S^t \in W_x^t$ with probability of at least $1-\negl(n)$.
Overall we get that for every turn $q_t$ Alice picks a number that has already been played with probability at least 

        $$\dfrac{2n}{2n+1} \cdot \dfrac{1}{2} - \negl(n)$$

In order not to repeat the same number twice and lose the game, Alice needs to guess the right set for $S^t$ for every $t = 1, 2, \ldots,\frac{c}{2}$ (in some order). But whenever she makes such a guess, her probability of succeeding given that she has succeeded so far is roughly $(n+1)/(2n)$. Therefore the probability she succeeds in all of the cases is at most 

 $$\left(  \dfrac{n+1}{2n+1} \right)^{c/2} +\negl(n) $$

\end{proof}

\begin{corollary}
If Alice has no secrets, and memory of size at most $\frac{n}{4c}$ bits (for $c\geq 2$), then there exists a strategy for Bob where for every $\epsilon > 0$ and large enough $n$ Alice loses with probability at least
$$1 - \left(\dfrac{1}{2} + \epsilon \right)^{c/2}.$$
\end{corollary}
\begin{proof}
The only requirement on $\delta$ in Theorem~\ref{theorem-alice-loses-whp} is that $2^{\delta}\beta < 1$. Therefore we can choose  $\delta$ s.t.:

    $$\delta \in \left(-\frac{3}{4}\log\beta,-\log\beta\right).$$

The only requirement on $\beta \in (0,1)$ is that $\beta > \alpha=(\frac{1}{c})^{1/4c}$. By choosing $\beta$ s.t.\
$\log\beta \in \left(\log\alpha,\frac{2}{3}\log\alpha \right)$ we get that
$\delta > -\frac{3}{4}\log\beta > -\frac{3}{4}\cdot\frac{2}{3}\log\alpha = -\frac{1}{2}\log\alpha$.
Therefore 
$$\delta >  -\frac{1}{2}\log\left(\left(\frac{1}{c}\right)^{1/4c}\right) = \frac{1}{8c}\log c > {\frac{1}{8c}}.$$ 

By Theorem~\ref{theorem-alice-loses-whp} the corollary  follows.  
\end{proof}

\section{Further Work}
Note that the strategy for Bob requires exponential time computation.
This suggests the open question of whether Bob's strategy can be made to be computationally efficient.  In particular, a possible direction for getting an efficient strategy  for Bob is to simulate Alice on various future selections, with the hope that this simulation reveals which set and matching she is actually storing in her memory and hence yield a move to play that will make her lose.    

Another question is whether it is possible to have a closed book Alice with little memory and only a few bits of long term randomness where even an all powerful Bob will not be able to  beat her. Recall that under computational limitations on Bob (the existence of a pseudorandom function that Bob cannot distinguish from random) this is possible, as pointed out by Menuhin and Naor~\cite{MenuhinN22}.

\noindent
{\bf Extending Eventown:}
Is there an extension of Berlekamp's Theorem that will tell us that given enough subsets not only are there two subsets whose union is of even cardinality but there are $c$ sets such that the union of any $c'$ of them is of even cardinality for all $c' \leq c$. This type of result would have made the amplification much simpler, as in Section~\ref{sec-half}, following turn $2k$ Bob finds $s_1, s_2, \ldots, s_c$ sets (where the actual used set $s$ is one of them, they all have similar conditional probability, i.e $\Pr[S=s_i|x,r_A]$ and such that the union of any $c'$ of them is of even cardinality). Bob's new strategy in turn t(for $t>2k$) is to pick some number that has not been used and
that is not in $\bigcup_{i\in A} s_{i}$, where $A\subseteq \{1,2...c\}$ indicate the indexes $i$ s.t.\ no number from $s_i$ has been used from turn 2k up to turn t. Since the union of any $c'$ subsets from $s_1,...,s_c$ is of even cardinality, we have that the first one to pick a number from $s_i$ after turn 2k is Alice. Thus, to secure a draw, Alice need to "guess" the right $s_i$, but this happens with a probability close to $1/c$ (instead of $1/2$). This also has the potential to yield a sharper result in terms of the probability of Alice drawing as a function of the memory she has.     

The latter raises the general question of whether  there is a better amplification procedure where you get that against an $o(n)$ memory Alice there is a strategy for Bob that make Bobs win with all but  exponentially in $n$  small probability. 

\noindent
{\bf Alice's side - precise bounds:} 
What is the precise bound on the size of memory Alice  actually needs in order to survive this ordeal. We can think of a deterministic strategy that requires $n$ bits (recall that there are $2n$ numbers all together). Alice fixes any perfect matching on the elements
$\{1, 2, \ldots, 2n\}$; for example $(1, 2), (3, 4), \ldots, (2n-1, n)$. Every bit $b_i$ represents a pair of numbers $(2i-1,2i)$. If $b_i=1$ it means that the numbers $2i-1,2i+1$ have both already been used and $b_i=0$ means that neither have been used. In addition  imagine that Alice stores an index $k$ of a value that was used but its partner was not used (we need to maintain the invariant that there is always at most one such value - in fact we will make sure it is in the smallest number has not be used so far, so there isn't even any need to store it explicitly). 

Therefore Alice's strategy is to start with the value $1$. Every time Bob chooses a number $j$ that is different from $2$ (the matched partner of $1$), Alice picks $j$'s matched number and changes the relevant bit to $1$. When Bob chooses $2$, then Alice set $b_1=1$ and finds, using the $n$ bits of memory, $b_1, b_2, \ldots, b_n$, the first pair that has not been played so far. She then chooses the first number from the pair as her move3e and sets $k$ accordingly. She continues in the same way till the end.

Recall that Garg and Schneider~\cite{mirrorgame} showed a lower bound on a deterministic Alice of around 0.664n bits. 
Is this tight? Perhaps an open book Alice can achieve better results.

\noindent
{\bf Limiting the memory size of both Alice and Bob:} 
What happens if {\em both} Alice and Bob have limited memory? In the spirit of the work on cooperation in repeated prisoner's dilemma with bounded rationality~\cite{PapadimitriouY94},   is there some sort of {\em stable solution} to the mirror game where it is possible to argue that it is not beneficial to neither party to deviate from a scripted solution (assuming they prefer winning over drawing and that over losing)?

\noindent 
{\bf Open book setting:} The open book setting appears in a variety of fields. In distributed literature it resembles the full information model (see \cite{DBLP:journals/siamcomp/GoldreichGL98}). In random sampling it similar to a fully adaptive adversary (see \cite{DBLP:conf/pods/Ben-EliezerY20}).
An interesting question is whether we can use similar techniques in order to find lower bounds in different multi players problems where one of the parties has no secrets.

A similar concern to open book appeared in the differential  privacy literature.  
Dwork et al.~\cite{DworkNPRY10}  considered a model of {\em pan-private streaming algorithms} where, as in the open book setting,
the entire internal state of the algorithm might be exposed to the adversary (either a fixed number of times or in what they called ``continual intrusions"). They construct
several streaming algorithms that remain secure in the pan-private model (e.g.\ distinct elements
count, incidence count, cropped mean, heavy hitters). 

\noindent
{\bf Adversarial streams and sampling:} Recently there has been a lot of exciting research on streaming in adversarial environments.  A streaming algorithm is called {\em adversarially robust} if
its performance holds even when the elements in the stream are chosen adaptively and in an adversarial manner after seeing the currents approximation or sample. The question is whether tasks that can be done with low memory against a stream that does not change as a function  of the seeing the current state can be done with low memory even if the status is public. For instance, on the positive side, Ben-Eliezer, Jayaram, Woodruff and Yogev~\cite{Ben-EliezerJWY20} showed general transformations for a family of tasks, for turning a streaming algorithm to be adversarially robust (with some overhead). On the other hand, Kaplan, Mansour, Nissim and Stemmer~\cite{KaplanMNS21} showed a problem that requires only polylogarithmic amount of memory in the static case but any adversarially
robust algorithm for it requires a much  larger memory.

The general question we ask is under what circumstance can we argue that a game where a player has no deterministic low memory strategy does not have an open book  (randomized) one as well. 

Independently of this work, Ajtai et al.~\cite{AjtaiBJSSWZ22} discussed a similar notion to ``Open Book" under the name ``White-Box Adversarial Data Stream Model" where the adversary sees the current state in full but does not know future coin flips. 
They showed a variety of results, for instance, a linear lower bound on the memory required for approximate frequency moment estimations in such a model.

\section*{Acknowledgements}
We thank Boaz Menuhin for very useful discussions and comments and thank Ehud Friedgut for discussions concerning finding an extension of Berlekamp's Theorem.  We thank the 
anonymous referees of FUN 2022 for numerous helpful comments.

\bibliography{main}

\begin{thebibliography}{10}

\bibitem{AjtaiBJSSWZ22}
Mikl{\'{o}}s Ajtai, Vladimir Braverman, T.~S. Jayram, Sandeep Silwal, Alec Sun,
  David~P. Woodruff, and Samson Zhou.
\newblock The white-box adversarial data stream model.
\newblock In {\em {PODS} '22: International Conference on Management of Data,
  Philadelphia, PA, USA, June 12 - 17, 2022}, pages 15--27. {ACM}, 2022.
\newblock \href {https://doi.org/10.1145/3517804.3526228}
  {\path{doi:10.1145/3517804.3526228}}.

\bibitem{Ben-EliezerJWY20}
Omri Ben{-}Eliezer, Rajesh Jayaram, David~P. Woodruff, and Eylon Yogev.
\newblock A framework for adversarially robust streaming algorithms.
\newblock In {\em Proceedings of the 39th {ACM} {SIGMOD-SIGACT-SIGAI} Symposium
  on Principles of Database Systems, {PODS} 2020, Portland, OR, USA, June
  14-19, 2020}, pages 63--80. {ACM}, 2020.
\newblock \href {https://doi.org/10.1145/3375395.3387658}
  {\path{doi:10.1145/3375395.3387658}}.

\bibitem{DBLP:conf/pods/Ben-EliezerY20}
Omri Ben{-}Eliezer and Eylon Yogev.
\newblock The adversarial robustness of sampling.
\newblock In {\em Proceedings of the 39th {ACM} {SIGMOD-SIGACT-SIGAI} Symposium
  on Principles of Database Systems, {PODS} 2020, Portland, OR, USA, June
  14-19, 2020}, pages 49--62. {ACM}, 2020.
\newblock \href {https://doi.org/10.1145/3375395.3387643}
  {\path{doi:10.1145/3375395.3387643}}.

\bibitem{DworkNPRY10}
Cynthia Dwork, Moni Naor, Toniann Pitassi, Guy~N. Rothblum, and Sergey
  Yekhanin.
\newblock Pan-private streaming algorithms.
\newblock In Andrew~Chi{-}Chih Yao, editor, {\em Innovations in Computer
  Science - {ICS} 2010, Tsinghua University, Beijing, China, January 5-7, 2010.
  Proceedings}, pages 66--80. Tsinghua University Press, 2010.
\newblock URL:
  \url{http://conference.iiis.tsinghua.edu.cn/ICS2010/content/papers/6.html}.

\bibitem{feige}
Uriel Feige.
\newblock A randomized strategy in the mirror game.
\newblock {\em arXiv preprint arXiv:1901.07809}, 2019.

\bibitem{fox-notes}
Jacob Fox.
\newblock Lecture notes for applications of linear algebra, mat 307,
  combinatorics.
\newblock URL: \url{https://math.mit.edu/~fox/MAT307-lecture15.pdf}.

\bibitem{mirrorgame}
Sumegha Garg and Jon Schneider.
\newblock The space complexity of mirror games.
\newblock In {\em 10th Innovations in Theoretical Computer Science Conference,
  {ITCS} 2019, January 10-12, 2019, San Diego, California, {USA}}, volume 124
  of {\em LIPIcs}, pages 36:1--36:14. Schloss Dagstuhl - Leibniz-Zentrum
  f{\"{u}}r Informatik, 2019.
\newblock \href {https://doi.org/10.4230/LIPIcs.ITCS.2019.36}
  {\path{doi:10.4230/LIPIcs.ITCS.2019.36}}.

\bibitem{DBLP:journals/siamcomp/GoldreichGL98}
Oded Goldreich, Shafi Goldwasser, and Nathan Linial.
\newblock Fault-tolerant computation in the full information model.
\newblock {\em {SIAM} J. Comput.}, 27(2):506--544, 1998.
\newblock \href {https://doi.org/10.1137/S0097539793246689}
  {\path{doi:10.1137/S0097539793246689}}.

\bibitem{KaplanMNS21}
Haim Kaplan, Yishay Mansour, Kobbi Nissim, and Uri Stemmer.
\newblock Separating adaptive streaming from oblivious streaming using the
  bounded storage model.
\newblock In {\em Advances in Cryptology - {CRYPTO} 2021, August 16-20, 2021,
  Proceedings, Part {III}}, volume 12827 of {\em Lecture Notes in Computer
  Science}, pages 94--121. Springer, 2021.
\newblock \href {https://doi.org/10.1007/978-3-030-84252-9\_4}
  {\path{doi:10.1007/978-3-030-84252-9\_4}}.

\bibitem{MenuhinN22}
Boaz Menuhin and Moni Naor.
\newblock Keep that card in mind: Card guessing with limited memory.
\newblock In {\em 13th Innovations in Theoretical Computer Science Conference,
  {ITCS} 2022, January 31 - February 3, 2022, Berkeley, CA, {USA}}, volume 215
  of {\em LIPIcs}, pages 107:1--107:28. Schloss Dagstuhl - Leibniz-Zentrum
  f{\"{u}}r Informatik, 2022.
\newblock \href {https://doi.org/10.4230/LIPIcs.ITCS.2022.107}
  {\path{doi:10.4230/LIPIcs.ITCS.2022.107}}.

\bibitem{PapadimitriouY94}
Christos~H. Papadimitriou and Mihalis Yannakakis.
\newblock On complexity as bounded rationality (extended abstract).
\newblock In Frank~Thomson Leighton and Michael~T. Goodrich, editors, {\em
  Proceedings of the Twenty-Sixth Annual {ACM} Symposium on Theory of
  Computing, 23-25 May 1994, Montr{\'{e}}al, Qu{\'{e}}bec, Canada}, pages
  726--733. {ACM}, 1994.
\newblock \href {https://doi.org/10.1145/195058.195445}
  {\path{doi:10.1145/195058.195445}}.

\end{thebibliography}

\end{document}